\documentclass{article}
\usepackage{ijcai20}

\usepackage{times}
\usepackage{soul}
\usepackage{url}
\usepackage[hidelinks]{hyperref}
\usepackage[utf8]{inputenc}
\usepackage[small]{caption}
\usepackage{graphicx}
\usepackage{amsmath}
\usepackage{amsthm}
\usepackage{booktabs}
\usepackage{algorithm}
\usepackage{algorithmic}

\usepackage{amssymb}
\usepackage{amsmath}
\usepackage{amsthm}
\usepackage{algorithm}
\usepackage{algorithmic}
\usepackage{subfig}
\usepackage{graphicx}
\usepackage{enumerate}

\makeatletter
\newcommand{\rmnum}[1]{\romannumeral #1}
\newcommand{\Rmnum}[1]{\expandafter\@slowromancap\romannumeral #1@}
\makeatother

\newtheorem{theorem}{Theorem}

\newtheorem{definition}{Definition}
\newtheorem{lemma}{Lemma}
\newtheorem{proposition}{Proposition}
\newtheorem{remark}{Remark}

\title{On Metric DBSCAN with Low Doubling Dimension}
\author{
Hu Ding$^1$
\and
Fan Yang$^1$
\affiliations
$^1$School of Computer Science and Technology, University of Science and Technology of China 
\emails
huding@ustc.edu.cn 
yang208@mail.ustc.edu.cn,
}
\begin{document}

\maketitle

\begin{abstract}
The density based clustering method {\em Density-Based Spatial Clustering of Applications with Noise (DBSCAN)} is a popular method for outlier recognition and has received tremendous attention from many different areas. A major issue of the original DBSCAN is that the time complexity could be as large as quadratic. Most of existing DBSCAN algorithms focus on developing efficient index structures to speed up the procedure in low-dimensional Euclidean space. However, the research of DBSCAN in high-dimensional Euclidean space or general metric space is still quite limited, to the best of our knowledge. In this paper, we consider the metric DBSCAN problem under the assumption that the inliers (excluding the outliers) have a low doubling dimension. We apply a novel randomized $k$-center clustering idea to reduce the complexity of range query, which is the most time consuming step in the whole DBSCAN procedure. Our proposed algorithms do not need to build any complicated data structures and are easy to be implemented in practice. The experimental results show that our algorithms can significantly outperform the existing DBSCAN algorithms in terms of running time. 
\end{abstract}

\section{Introduction}
\label{sec-intro}

Density-based clustering is a fundamental topic in data analysis and has many applications in the areas, such as machine learning, data mining, and computer vision~\cite{tan2006introduction}. Roughly speaking, the problem of density-based clustering aims to partition given data set into clusters where each cluster is a dense region in the space. The remaining data located in sparse regions are recognized as ``outliers''. Note that the given data set can be a set of points in a Euclidean space or any abstract metric space. {\em DBSCAN (Density-Based Spatial Clustering of Applications with Noise )}~\cite{ester1996density} is one of the most popular density-based clustering methods and has been implemented for solving many real-world problems. DBSCAN uses two parameters, ``$MinPts\geq 1$'' and ``$\epsilon>0$'', to define the clusters ({\em i.e.,} the dense regions): a point $p$ is a ``core point'' if it has at least $MinPts$ neighbors within distance $\epsilon$; a cluster is formed by a set of ``connected'' core points and some non-core points located in the boundary (which are named ``border points''). We will provide the formal definition in Section~\ref{sec-dbscandef}.

A bottleneck of the original DBSCAN algorithm is that it needs to perform a range query for each data item, {\em i.e.,} computing the number of neighbors within the distance $\epsilon$, and the overall time complexity can be as large as $O(n^2\beta)$ in the worst case, where $n$ is the number of data items and $\beta$ indicates the complexity for computing the distance between two items. For example, if the given data is a set of points in $\mathbb{R}^D$, we have $\beta=O(D)$.
When $n$ or $\beta$ is large, the procedure of range query could make DBSCAN running very slowly. 

Most existing DBSCAN algorithms focus on the case in low-dimensional Euclidean space. To speed up the step of range query, a natural idea is using some efficient index structures, such as $R^*$-tree~\cite{DBLP:conf/sigmod/BeckmannKSS90},  though the overall complexity in the worst case is still $O(n^2)$ ($\beta=O(1)$ for low-dimensional Euclidean space).  We refer the reader to the recent articles that systematically discussed this issue~\cite{gan2015dbscan,schubert2017dbscan}. 

Using the novel techniques from computational geometry, the running time of DBSCAN in $\mathbb{R}^2$ has been improved from $O(n^2)$ to $O(n\log n)$ by~\cite{DBLP:conf/isaac/BergGR17,gunawan2013faster}. For the case in general $D$-dimensional Euclidean space, \cite{DBLP:journals/ijcga/ChenSX05} and \cite{gan2015dbscan} respectively provided the algorithms achieving sub-quadratic running times, where their complexities are both in the form of $O(n^{f(D)}\cdot D)$ with $f(D)$ being some function satisfying $\lim_{D\rightarrow\infty}f(D)=2$. Namely, when the dimensionality $D$ is high, their algorithms cannot gain a significant improvement over the straightforward implementation that has the complexity $O(n^2D)$. Recently, \cite{DBLP:conf/icml/JangJ19} proposed a sampling based method, called {\em DBSCAN++}, to compute an approximate solution for DBSCAN; but their sample size $m\approx n$ when the dimensionality $D$ is large (so there is no significant difference in terms of the time complexity if running the DBSCAN algorithm on the sample). 

To speed up DBSCAN in practice, a number of approximate and distributed DBSCAN algorithms have been proposed, such as \cite{gan2015dbscan,yang2019dbscan,lulli2016ng,song2018rp,DBLP:conf/icml/JangJ19}. 
To the best of our knowledge, most of these algorithms only consider the instances in low-dimensional Euclidean space (rather than high-dimensional Euclidean space or abstract metric space), except \cite{lulli2016ng,yang2019dbscan}. \citeauthor{lulli2016ng} presented an approximate, distributed algorithm for DBSCAN, as long as the distance function $d(\cdot, \cdot)$ is symmetric, that is, $d(x, y)=d(y, x)$ for any two points $x$ and $y$. Very recently, \citeauthor{yang2019dbscan} showed an exact, distributed algorithm for DBSCAN in any abstract metric space; however, their method mainly focuses on how to ensure the load balancing and cut down the communication cost for distributed systems, rather than reducing the computational complexity of DBSCAN (actually, it directly uses the original DBSCAN algorithm of \citeauthor{ester1996density} in each local machine).

\subsection{Our Main Results}
\label{sec-our}
In this paper, we consider developing efficient algorithm for computing the exact solution of DBSCAN. 
As mentioned by \citeauthor{yang2019dbscan}, a wide range of real-world data cannot be represented  in low-dimensional Euclidean space ({\em e.g.,} textual and image data can only be embedded into high-dimensional Euclidean space).
Moreover, as mentioned in ~\cite{schubert2017dbscan}, the original DBSCAN was designed for general metrics, as long as the distance function of data items can be well defined. Thus it motivates us to consider the problem of DBSCAN in high-dimensional Euclidean space and general metric space. 

We assume that the given data has a low ``doubling dimension'', which is widely used for measuring the intrinsic dimensions of datasets~\cite{talwar2004bypassing} (we provide the formal definition in Section~\ref{sec-dmdef}).  
The rationale behind this assumption is that many real-world datasets manifest low intrinsic dimensions~\cite{belkin2003problems}. 
For example, image sets usually can be represented in low dimensional manifold though the Euclidean dimension of the image vectors can be very high. We also note that it might be too strict to assume that the whole data set has a low doubling dimension, especially when it contains outliers. To make the assumption more general and capture a broader range of cases in practice, we only assume that the set of inliers has a constant doubling dimension while the outliers can scatter arbitrarily in the space. The assumption is formally stated in Definition~\ref{def-assumption}. 
%
We focus on the following key question: 

\vspace{0.05in}
{\em Is there any efficient algorithm being able to reduce the complexity of range query for DBSCAN, under such ``low doubling dimension assumption''?}
\vspace{0.05in}

We are aware of several index structures in doubling metrics, {\em e.g.,}~\cite{karger2002finding,krauthgamer2004navigating,talwar2004bypassing}. However, these methods cannot handle the case with outliers. Moreover, they usually need to build very complicated data structures ({\em e.g.,} hierarchically well-separated tree) that are not quite efficient in practice. 

We observe that the well-known $k$-center clustering procedure can be incorporated into the DBSCAN algorithm to reduce the complexity of the range query procedure in doubling metric. However, we cannot directly apply the ordinary $k$-center clustering method ({\em e.g.,} \cite{gonzalez1985clustering}) since the outliers may not satisfy the low doubling dimension condition. Instead, we show that a randomized $k$-center clustering algorithm proposed by \cite{DYW} can efficiently remedy this issue, though we still need to develop some new ideas to apply their algorithm to solve the problem of DBSCAN.

The rest of the paper is organized as follows. In Section~\ref{sec-pre}, we show the formal definitions of doubling dimension and DBSCAN, and briefly introduce the randomized $k$-center clustering algorithm from~\cite{DYW}. In Section~\ref{sec-alg}, we propose and analyze our algorithms for reducing the complexity of range query in detail. Finally, we compare the experimental performances of our algorithms and several well-known baseline DBSCAN algorithms on both synthetic and real datasets.

%


%

\section{Preliminaries}
\label{sec-pre}
Throughout this paper, we use $(X, d)$ to denote the metric space where $d(\cdot, \cdot)$ is the distance function on $X$. Let $|X|=n$. We also assume that it takes $O(\beta)$ time to compute $d(p, q)$ for any $p$, $q\in X$. 
Let $Ball(x, r)$ be the ball centered at point $x\in X$ with radius $r\geq 0$ in the metric space.

\subsection{DBSCAN}
\label{sec-dbscandef}
 
We introduce the formal definition of DBSCAN. 
Given two parameters $\epsilon>0$ and $MinPts\in\mathbb{Z}^+$, DBSCAN divides the points of $X$ into three classes: 

\begin{enumerate}
\item $p$ is a \textbf{core point}, if $|Ball(p, \epsilon)\cap X|\geq MinPts$;
\item $p$ is a \textbf{border point}, if $p$ is not a core point but $p\in Ball(q, \epsilon)$ of some core point $q$;
\item all the other points are \textbf{outliers}.
\end{enumerate}

To define a cluster of DBSCAN, we need the following concept.

\begin{definition}[Density-reachable]
\label{def-reach}
We say a point $p\in X$ is density-reachable from a core point $q$, if there exists a sequence of points $p_1, p_2, \cdots, p_t\in X$ such that:

\begin{itemize}
\item $p_1=q$ and $p_t=p$;
\item $p_1, \cdots, p_{t-1}$ are all core points;
\item $p_{i+1}\in Ball(p_i, \epsilon)$ for each $i=1, 2, \cdots, t-1$.
\end{itemize}
\end{definition}

If one arbitrarily picks a core point $q$, then the corresponding DBSCAN cluster defined by $q$ is
\begin{eqnarray}
\{p\mid p\in X  \text{ and p  is density-reachable from q}\}.\label{for-dbscan1}
\end{eqnarray}
Actually, we can imagine that the set $X$ form a directed graph: any two points $p$ and $p'\in P$ are connected by a directed edge $p\rightarrow p'$, if $p$ is a core point and $p'\in Ball(p, \epsilon)$. From (\ref{for-dbscan1}), we know that the cluster is the maximal subset containing the points who are density-reachable from $q$. The cluster may contain both core and border points. It is easy to know that for any two core point $q$ and $q'$, they define exactly the same cluster if they are density-reachable from each other ({\em i.e.,} there exists a path from $q$ to $q'$ and vice versa). Therefore, a cluster of DBSCAN is uniquely defined by any of its core points. Moreover, a border point could belong to multiple clusters and an outlier cannot belong to any cluster. 
The goal of DBSCAN is to discover these clusters and outliers. 

For convenience, we use $X_{in}$ and $X_{out}$ to denote the sets of inliers (including the core points and border points) and outliers, respectively.

\subsection{Doubling Metrics}
\label{sec-dmdef}


\begin{definition}[Doubling Dimension]
\label{def-dd}
The doubling dimension of a metric $(X,d)$ is the smallest number $\rho>0$, such that for any $p\in X$ and $r\geq 0$, $X\cap Ball(p, 2r)$ is always covered by the union of at most $2^\rho$ balls with radius $r$.
\end{definition}

Roughly speaking, the doubling dimension describes the expansion rate of the metric. 
We have the following property of doubling metrics from~\cite{talwar2004bypassing,krauthgamer2004navigating} that can be proved by recursively applying Definition~\ref{def-dd}.

\begin{proposition}
\label{pro-doubling}
Let $(X, d)$ be a metric with the doubling dimension $\rho>0$. If $Y\subseteq X$ and its aspect ratio is $\alpha =\frac{\max_{y, y'\in Y}d(y, y')}{\min_{y, y'\in Y}d(y, y')}$, then $|Y|\leq 2^{\rho\lceil\log\alpha\rceil}$.
\end{proposition}

For our DBSCAN problem, we adopt the following assumption from~\cite{DYW}.


\begin{definition}[Low Doubling Dimension Assumption]
\label{def-assumption}
Given an instance $(X, \epsilon, MinPts)$ of DBSCAN, we assume that the metric $(X_{in}, d)$, {\em i.e.,} the metric formed by the set of core points and border points, has a constant doubling dimension $\rho>0$. The set $X_{out}$ of outliers can scatter arbitrarily in the space.
\end{definition}

\subsection{The Randomized Gonzalez's Algorithm}
\label{sec-rangon}
{\em $k$-center clustering} is one of the most fundamental clustering problems~\cite{gonzalez1985clustering}. Given a metric $(X, d)$ with $|X|=n$, the problem of $k$-center clustering is to find $k$ balls to cover the whole $X$ and minimize the maximum radius. 
For the sake of completeness, let us briefly introduce the algorithm of \cite{gonzalez1985clustering} for $k$-center clustering first. Initially, it arbitrarily selects a point from $X$, and iteratively selects the following $k-1$ points, where each $j$-th step ($2\leq j\leq k$) chooses the point having the largest minimum distance to the already selected $j-1$ points; finally, each point of $X$ is assigned to its nearest neighbor of these selected $k$ points. It can be proved that this greedy strategy yields a $2$-approximation of $k$-center clustering ({\em i.e.,} the maximum radius of the obtained $k$ balls is at most twice as large as the optimal radius). 

\cite{DYW} presented a randomized version of the Gonzalez's algorithm for solving $k$-center clustering with outliers. Let $z\geq 1$ be the pre-specified number of outliers, and the problem of $k$-center with outliers is to find $k$ balls to cover $n-z$ points of $X$ and minimize the maximum radius. This problem is much more challenging than the ordinary $k$-center clustering, since we do not know which points are the outliers in advance and there are an exponentially large number ${n\choose z}$ of  different possible cases. Note that other algorithms for $k$-center clustering with outliers, such as \cite{charikar2001algorithms,DBLP:conf/icalp/ChakrabartyGK16}, take at least quadratic time complexity. The key idea of \cite{DYW} is to replace each step of Gonzalez's algorithm, choosing the farthest point to the set of already selected points, by taking a random sample from the farthest $(1+\delta)z$ points, where $\delta>0$ is a small parameter; after $O(k)$ steps, with constant probability, the algorithm yields a set of $O(\frac{k}{\delta})$ balls covering at least $n-(1+\delta)z$ points of $X$ and the resulting radius is at most twice as large as the optimal radius.  For example, if we set $\delta=1$, the algorithm will yield $O(k)$ balls covering at least $n-2z$ points. The formal result is presented in Theorem~\ref{the-bi}. 
We omit the detailed proof from \cite{DYW}. 

\begin{algorithm}[tb]
   \caption{The Randomized Gonzalez's algorithm}
   \label{alg-bi}
\begin{algorithmic}
  \STATE {\bfseries Input:} An instance $(X, d)$ of $k$-center clustering with $z$ outliers, and $|X|=n$; the parameters $\delta>0$, $\eta\in (0,1)$, and $t\in \mathbb{Z}^+$.
   \STATE
\begin{enumerate}
\item Let $\gamma=z/n$ and initialize a set $E=\emptyset$. 

\item Initially, $j=1$; randomly select $\frac{1}{1-\gamma}\log\frac{1}{\eta}$ points from $X$ and add them to $E$.
\item Run the following steps until $j= t$:
\begin{enumerate}
\item $j=j+1$ and let $Q_j$ be the farthest $(1+\delta)z$ points of $X$ to $E$ (for each point $p\in X$, its distance to $E$ is $\min_{q\in E}d(p, q)$). 
\item Randomly select $\frac{1+\delta}{\delta}\log\frac{1}{\eta}$ points from $Q_j$ and add them to $E$.
\end{enumerate}
\end{enumerate}
  \STATE {\bfseries Output} $E$.
\end{algorithmic}
\end{algorithm}

\begin{theorem}
\label{the-bi}
Let $r_{opt}$ be the optimal radius of the instance $(X, d)$ of $k$-center clustering with $z$ outliers. If we set $t=\frac{k+\sqrt{k}}{1-\eta}$ in  Algorithm~\ref{alg-bi}, with probability at least $(1-\eta)(1-e^{-\frac{1-\eta}{4}})$, the set of balls
\begin{eqnarray}
\cup_{c\in E}Ball\big(c, 2r_{opt}\big)
\end{eqnarray}
cover at least $n-(1+\delta)z$ points of $X$.
\end{theorem}
If $\frac{1}{\eta}$ and $\frac{1}{1-\gamma}$ are constant numbers, the number of balls ({\em i.e.,} $|E|$) is $O(\frac{k}{\delta})$ and the success probability is constant. In each round of Step~3, there are $\frac{1+\delta}{\delta}\log\frac{1}{\eta}=O(\frac{1}{\delta})$ new points added to $E$, thus it takes $O(\frac{1}{\delta}n\beta)$ time to update the distances from the points of $X$ to $E$; to select the set $Q_j$, we can apply the linear time selection algorithm~\cite{blum1973time}. Overall, the running time of Algorithm~\ref{alg-bi} is $O(\frac{k}{\delta}n\beta)$. If the given instance is in $\mathbb{R}^D$, the running time will be $O(\frac{k}{\delta}n D)$.



\section{Our Algorithms and Theoretical Analysis}
\label{sec-alg}

In this section, we present two efficient algorithms for solving DBSCAN under the assumption of Definition~\ref{def-assumption}.

\subsection{The First DBSCAN Algorithm}
\label{sec-dbscan1}

Our first DBSCAN algorithm (Algorithm~\ref{alg-alg1}) contains two parts. To better understand our algorithm, we briefly introduce the high-level idea below.  For convenience, we use $d(U, V)$ to denote the minimum distance between  two sets $U$ and $V\subset X$, {\em i.e.,} $\min\{d(u, v)\mid u\in U, v\in V\}$. 

\vspace{0.03in}
\textbf{Part (\rmnum{1}).} First, we run Algorithm~\ref{alg-bi} to conduct a coarse partition on the given set $X$.  We view $X$ as an instance of $k$-center clustering with $z$ outliers where $z=|X_{out}|$ (recall $X_{out}$ is the set of outliers defined in Section~\ref{sec-dbscandef}). However, we cannot directly run Algorithm~\ref{alg-bi} since the values of $t$ and $z$ are not given. Actually, we can avoid to set the value of $t$ via a slight modification on Algorithm~\ref{alg-bi}; we just need to iteratively run Step 3 until $d(Q_j, E)\leq r$, where $r$ is a parameter that will be discussed in our experiments. For the parameter $z$, we cannot obtain its exact value before running DBSCAN; so we only assume that an upper bound $\tilde{z}$ of $z$ is available in our experiments. In practice, the number $z$ is much smaller than $n$. In each round of Step 3 of Algorithm~\ref{alg-bi}, we update the distances from $X$ to $E$. As a by-product,  we can store the following informations after running Algorithm~\ref{alg-bi}:
\begin{itemize}
\item the pairwise distances of $E$: $\{d(c, c')\mid c, c'\in E\}$;
\item for each $p\in X$, denote by $c_p$ its nearest neighbor in $E$.
\end{itemize}


%
%
If we simply set $\delta=1$ in Algorithm~\ref{alg-bi}, Theorem~\ref{the-bi} implies that at least $n-2\tilde{z}$ points of $X$ are covered by the balls $\cup_{c\in E}Ball\big(c, r\big)$. We denote the set of points outside the balls as $X_{\tilde{z}}$, and obviously $|X_{\tilde{z}}|$ is no larger than $2\tilde{z}$ by Theorem~\ref{the-bi}. 
\vspace{0.03in}

\textbf{Part (\rmnum{2}).} For the second part, we check each point $p\in X$ and determine its label to be ``core point'', ``border point'', or ``outlier''. According to the formulation of DBSCAN, we need to compute the size $\big|X\cap Ball(p, \epsilon)\big|$. 
In general, this procedure will take $O(n\beta)$ time and the whole running time will be $O(n^2 \beta)$. To reduce the time complexity, we can take advantage of the informations obtained in Part (\rmnum{1}). Since $|X_{\tilde{z}}|\leq 2\tilde{z}$ and $\tilde{z}$ usually is much smaller than $n$, we focus on the part $X\setminus X_{\tilde{z}}$ containing the majority of the points in $X$. Let $p$ be any point in $X\setminus X_{\tilde{z}}$ and $c_p$ be its nearest neighbor in $E$. Let 
\begin{eqnarray}
A_p=\{c\mid c\in E, d(c, c_p)\leq 2r+\epsilon\},
\end{eqnarray}
and we can quickly obtain the set $A_p$ since the pairwise distances of $E$ are stored in Part (\rmnum{1}). Lemma~\ref{lem-neighbor} guarantees that we only need to check the local region, the balls $\bigcup_{c\in A_p}  Ball(c, r )$ and $X_{\tilde{z}}$, instead of the whole $X$, for computing the size $\big|X\cap Ball(p, \epsilon)\big|$; further, Lemma~\ref{lem-ball} shows that the size of $A_p$ is bounded. See Figure~\ref{fig-alg} for an illustration.


\begin{algorithm}[tb]
   \caption{\sc{Metric DBSCAN Algorithm}}
   \label{alg-alg1}
\begin{algorithmic}
  \STATE {\bfseries Input:} An instance $(X, d)$ of DBSCAN, and the parameters $\epsilon, r>0$, $MinPts, \tilde{z}\in\mathbb{Z}^+$.
   \STATE
   \begin{enumerate}
   \item Run Algorithm~\ref{alg-bi} with setting $\delta=1$, and terminate the loop of Step 3 when $d(Q_j, E)\leq r$. 
   \begin{enumerate}
   
   \item Store the set $\mathcal{D}_E=\{d(c, c')\mid c, c'\in E\}$.
   \item For each $p\in X$, denote by $c_p$ its nearest neighbor in $E$.
   \item If the instance is in Euclidean space: for each $c\in E$ we build a $R^*$-tree for the points inside $Ball(c, r )$ (if a point $p$ is covered by multiple balls, we assign it to the ball of the center $c_p$).    
      \end{enumerate}
   \item For each $p\in X$, check whether it is a core point:
   \begin{enumerate}
   \item if $p\in X_{\tilde{z}}$, directly compute the set $X\cap Ball(p, \epsilon)$ by scanning $X$;
   \item else, obtain the set $A_p=\{c\mid c\in E, d(c, c_p)\leq 2r+\epsilon\}$ from $\mathcal{D}_E$, and compute the set $X\cap Ball(p, \epsilon)$ by checking the points in $\Big(\bigcup_{c\in A_p}\big(X\cap Ball(c, r )\big)\Big)\bigcup X_{\tilde{z}}$ (inside each $Ball(c, r )$, we use the $R^*$-tree built in Step 1(c) if the instance is in Euclidean space). 
      \end{enumerate}
      \item Join the core points into clusters by running the standard DBSCAN procedure~\cite{schubert2017dbscan}. 
   \end{enumerate}
%
\end{algorithmic}
\end{algorithm}

\newcounter{sd}
\begin{figure*}[]
   \centering
  \includegraphics[height=1.3in]{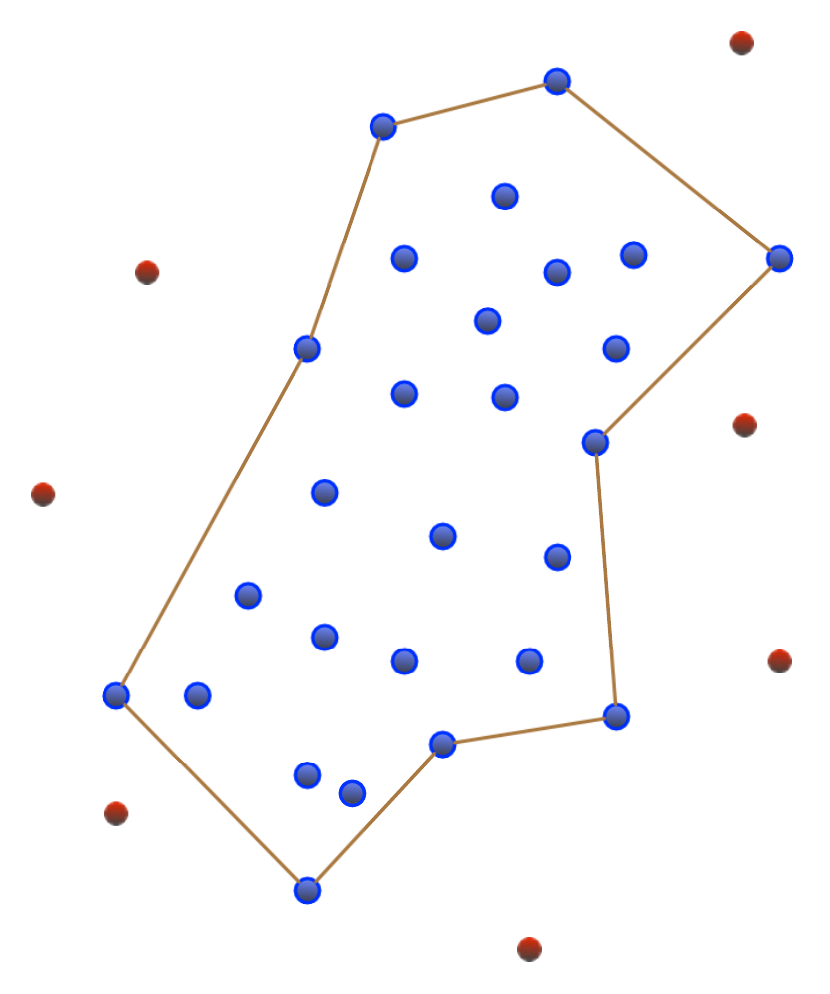}
  \hspace{0.8in}
    \includegraphics[height=1.3in]{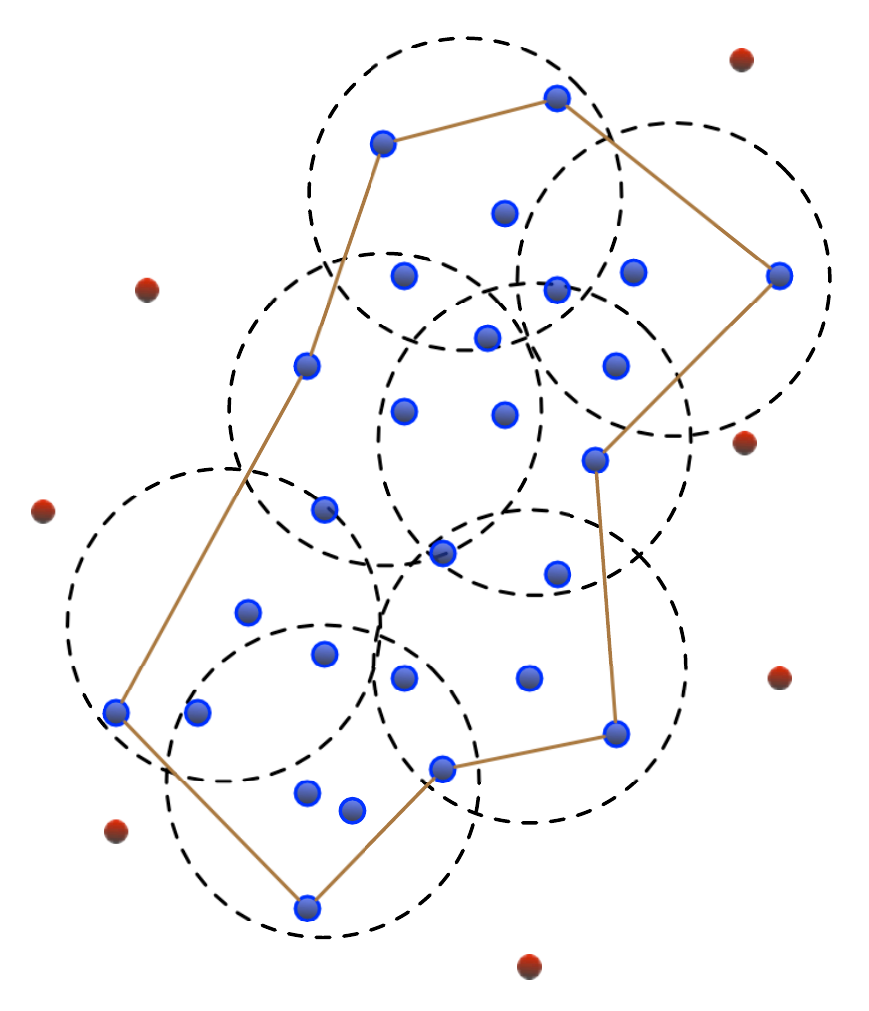} 
     \hspace{0.8in}
      \includegraphics[height=1.3in]{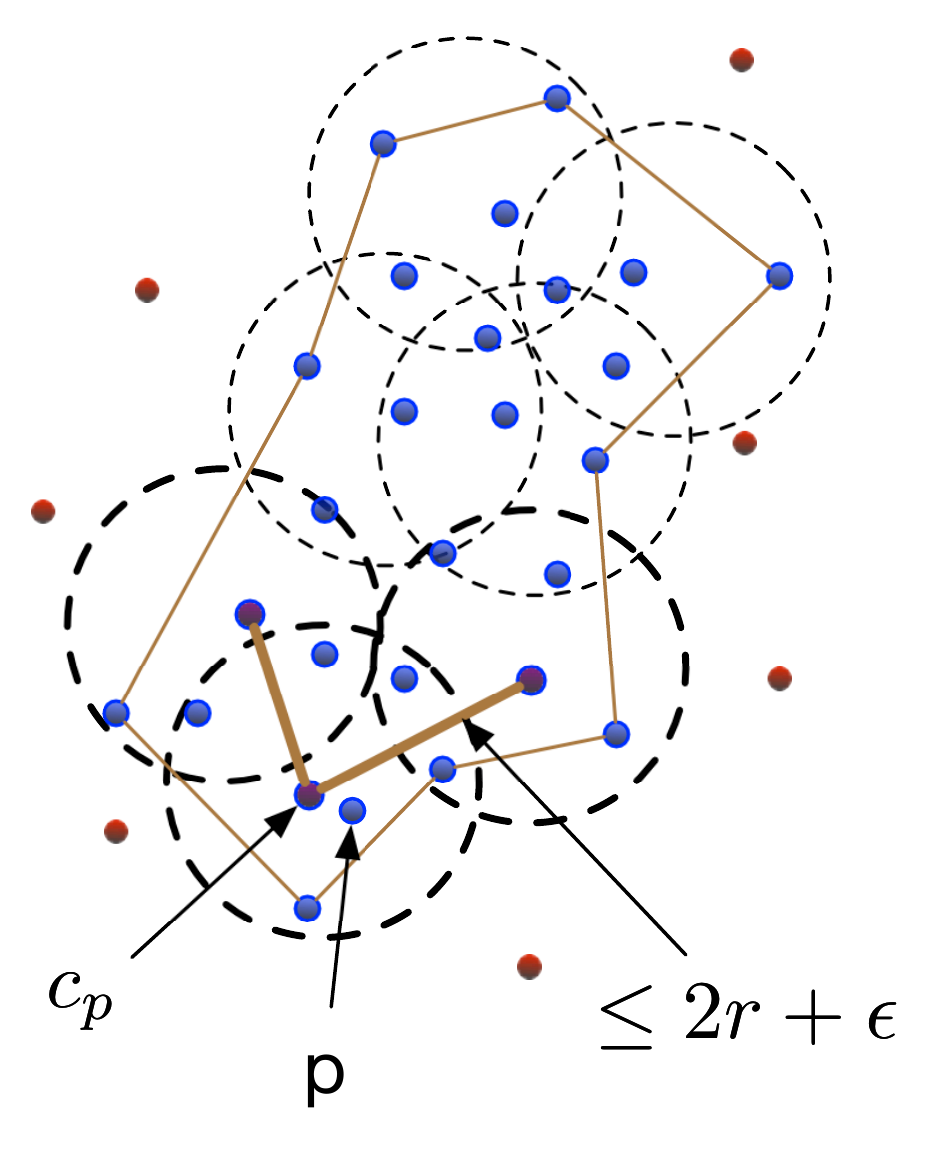}
      \centerline{ \hspace{-1in}\stepcounter{sd} \hfill (\alph{sd}) \hfill \stepcounter{sd} (\alph{sd}) \hfill \stepcounter{sd} (\alph{sd}) \hspace{1.2in}}
      \caption{(a) indicates an instance of DBSCAN, where the blue points are inliers (including the core points and border points) and the red points are outliers; (b) shows the balls obtained in Algorithm~\ref{alg-bi}; (c) shows an example of computing the set $X\cap Ball(p, \epsilon)$ for a point $p$, where we just need to check the two neighbor balls of $c_p$.}
  \label{fig-alg}
\end{figure*}


%
%

\begin{lemma}
\label{lem-neighbor}
If $p\in X\setminus X_{\tilde{z}}$, we have $X\cap Ball(p, \epsilon)\subset \Big(\bigcup_{c\in A_p}\big(X\cap Ball(c, r )\big)\Big)\bigcup X_{\tilde{z}}$.
\end{lemma}
\begin{proof}
Let $q$ be any point in $X\setminus X_{\tilde{z}}$. If $d(c_p, c_q)>2r+\epsilon$, {\em i.e.,} $q\in \bigcup_{c\notin A_p}Ball\big(c, r\big)$, by using the triangle inequality, we have
\begin{eqnarray}
d(p, q)&\geq& d(c_p, c_q)-d(p, c_p)-d(q, c_q)\nonumber\\
&>&2r+\epsilon-r-r>\epsilon.
\end{eqnarray}
Therefore, $q\notin X\cap Ball(p, \epsilon)$. That is, 
\begin{eqnarray}
X\cap Ball(p, \epsilon)&\subset& X\setminus \Big(\bigcup_{c\notin A_p}Ball\big(c, r\big)\Big)\nonumber\\
&=&\Big(\bigcup_{c\in A_p}\big(X\cap Ball(c, r )\big)\Big)\bigcup X_{\tilde{z}}.
\end{eqnarray}
So we complete the proof.
\end{proof}

Now, we consider the size of $A_p$. 
Recall the construction process of $E$ in Algorithm~\ref{alg-bi}. 
Initially, Algorithm~\ref{alg-bi} adds $\frac{1}{1-\gamma}\log\frac{1}{\eta}$ points to $E$; in each round of Step 3, it adds $2\log\frac{1}{\eta}$ points to $E$ (since we set $\delta=1$). So we can imagine that $E$ consists of multiple ``batches'' where each batch contains $\leq\max\{\frac{1}{1-\gamma}\log\frac{1}{\eta},2\log\frac{1}{\eta}\}$ points. Also, since we terminate Step 3 when $d(Q_j, E)\leq r$, any two points from different batches should have distance at least $r$. We  consider the batches having non-empty intersection with $A_p$. For ease of presentation, we denote these batches as $B_1, B_2, \cdots, B_m$. Further,  we label each batch $B_j$ by two colors for $1\leq j\leq m$: 
\begin{itemize}
\item ``red'' if $B_j\cap A_p\cap X_{in}\neq \emptyset$;
\item ``blue'' otherwise.
\end{itemize}
Recall $X_{in}$ is the set of core points and border points defined in Section~\ref{sec-dbscandef}.  Without loss of generality, we assume that the batches $\{B_j\mid 1\leq j\leq m'\}$ are red, and the batches $\{B_j\mid m'+1\leq j\leq m\}$ are blue. 
To bound the size of $A_p$, we divide it to two parts $A_p\setminus X_{in}$ and $A_p\cap X_{in}$. It is easy to know that $A_p\setminus X_{in}\subset X_{out}$, {\em i.e.,}
\begin{eqnarray}
|A_p\setminus X_{in}|\leq |X_{out}|\leq \tilde{z}.
\end{eqnarray}
Also, $A_p\cap X_{in}$ belongs to the union of the red batches, and therefore
\begin{eqnarray}
|A_p\cap X_{in}|\leq m'\times(\max\{2, \frac{1}{1-\gamma}\}\cdot\log \frac{1}{\eta})
\end{eqnarray}
So we focus on the value of $m'$ below. 

\begin{lemma}
\label{lem-ball}
 The number of red batches, $m'$, is at most $2^{\rho\lceil\log\alpha\rceil}$, where $\alpha=4+2\frac{\epsilon}{r}$. That is, $|A_p\cap X_{in}|\leq 2^{\rho\lceil\log\alpha\rceil}\times(\max\{2, \frac{1}{1-\gamma}\}\cdot\log \frac{1}{\eta})$. For simplicity, if we assume $\frac{1}{\eta}$ and $\frac{1}{1-\gamma}$ are constant numbers in Algorithm~\ref{alg-bi}, then 
 $$|A_p\cap X_{in}|\leq O(2^{\rho\lceil\log\alpha\rceil}).$$
  \end{lemma}
\begin{proof}

For each red batch $B_j$, we arbitrarily pick one point, say $c_j$, from $B_j\cap A_p\cap X_{in}$, and let 
\begin{eqnarray}
H=\{c_j\mid 1\leq j\leq m'\}. 
\end{eqnarray}
First, we know $H\subset X_{in}$. Second, because the minimum pairwise distance of $H$ is at least $r$ (since any two points of $H$ come from different batches) and the maximum pairwise distance of $H$
\begin{eqnarray}
\max_{c_j, c_{j'}\in H} d(c_j, c_{j'})&\leq& \max_{c_j, c_{j'}\in H} \big(d(c_j, c_p)+d(c_p, c_{j'})\big)\nonumber\\
&\leq& 2(2r+\epsilon)=4r+2\epsilon,
\end{eqnarray}
the aspect ratio of $H$ is no larger than $\alpha=4+2\frac{\epsilon}{r}$. Note the doubling dimension of $(X_{in}, d)$ is $\rho$ according to Definition~\ref{def-assumption}. Through Proposition~\ref{pro-doubling}, we have $|H|\leq 2^{\rho\lceil\log\alpha\rceil}$.

So the number of red batches $m'=|H|\leq 2^{\rho\lceil\log\alpha\rceil}$; each batch has size $\leq\max\{2, \frac{1}{1-\gamma}\}\cdot\log \frac{1}{\eta}$. 
Overall, we have $|A_p\cap X_{in}|\leq 2^{\rho\lceil\log\alpha\rceil}\times(\max\{2, \frac{1}{1-\gamma}\}\cdot\log \frac{1}{\eta})$. 
\end{proof}

\subsection{An Alternative Approach}
\label{sec-dbscan2}

In this section, we provide a modified version of our first DBSCAN algorithm. In Lemma~\ref{lem-ball}, we cannot directly use Proposition~\ref{pro-doubling} to bound the size of $A_p$, because the points inside the same batch could have pairwise distance less than $r$; therefore, we can only bound the number of red batches. To remedy this issue, we perform the following ``filtration'' operation when adding each batch to $E$ in Algorithm~\ref{alg-bi}.

\textbf{Filtration.} For each batch of $E$, we compute a connection graph: each point of the batch represents a vertex, and any two vertices are connected by an edge if their pairwise distance is smaller than $r$. Then, we compute a maximal independent set (not necessary the maximum independent set) of the graph, and only add this independent set to $E$ instead of the whole batch. See Figure~\ref{fig-alg2} as an illustration.

\begin{figure}[]
   \centering
  \includegraphics[height=1in]{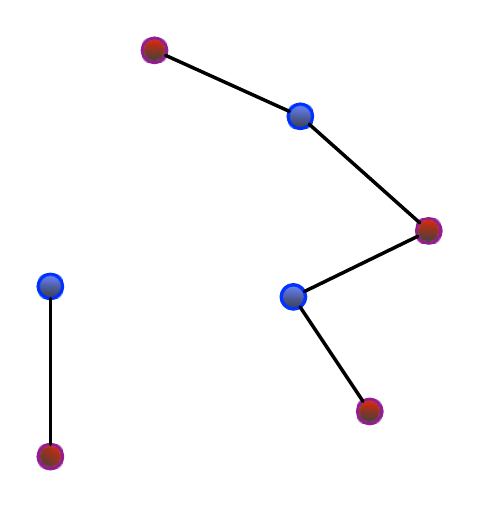}
      \caption{The batch contains $7$ points, and any two points are connected by an edge if their distance is smaller than $r$; we can pick the $4$ red points as the maximal independent set.}
  \label{fig-alg2}
\end{figure}

Obviously, this filtration operation guarantees that the pairwise distance of any two points in $E$ is at least $r$. Since each batch has size $\max\{2, \frac{1}{1-\gamma}\}\cdot\log \frac{1}{\eta}$, it takes $O\big((\max\{2, \frac{1}{1-\gamma}\}\cdot\log \frac{1}{\eta})^2\beta\big)$ time to compute the maximal independent set. Moreover, since the set $E$ has fewer points, we need to modify the result stated in Theorem~\ref{the-bi}. Let $p$ be any point of $X$ having distance no larger than $r$ to $E$ in the original Algorithm~\ref{alg-bi}. After performing the filtration operation, we know $d(p, E)\leq 2r$ due to the triangle inequality. As a consequence, the set $X\setminus X_{\tilde{z}}$ is covered by the balls $\cup_{c\in E}Ball\big(c, 2r\big)$ (instead of $\cup_{c\in E}Ball\big(c, r\big)$). Let 
\begin{eqnarray}
A'_p=\{c\mid c\in E, d(c, c_p)\leq 4r+\epsilon\}.
\end{eqnarray}
The aspect ratio of $A'_p$ is no larger than $\frac{2(4r+\epsilon)}{r}=8+2\frac{\epsilon}{r}$. Using the similar ideas for proving Lemma~\ref{lem-neighbor} and \ref{lem-ball}, we obtain the following results.

\begin{lemma}
\label{lem-neighbor2}
If $p\in X\setminus X_{\tilde{z}}$, we have $X\cap Ball(p, \epsilon)\subset \Big(\bigcup_{c\in A'_p}\big(X\cap Ball(c, 2r )\big)\Big)\bigcup X_{\tilde{z}}$.
\end{lemma}

\begin{lemma}
\label{lem-ball2}
%
$|A'_p\setminus X_{in}| \leq \tilde{z}$ and $|A'_p\cap X_{in}|\leq 2^{\rho\lceil\log\alpha\rceil}$, where $\alpha=8+2\frac{\epsilon}{r}$.
\end{lemma}
\begin{remark}
Comparing with the size $|A_p\cap X_{in}|$ in Lemma~\ref{lem-ball}, we remove the hidden constant ``$(\max\{2, \frac{1}{1-\gamma}\}\cdot\log \frac{1}{\eta})$'' but increase the value of $\alpha$ from $4+2\frac{\epsilon}{r}$ to $8+2\frac{\epsilon}{r}$.  So, we cannot directly compare the sizes $|A_p\cap X_{in}|$ and $|A'_p\cap X_{in}|$ in general. In Section~\ref{sec-exp}, we implement the two algorithms, and investigate their experimental performances. 
\end{remark}

%

\begin{figure*}[]
	\centering
	\subfloat[\textsc{Synthetic} ($r=100$)]{\includegraphics[height=1.2in]{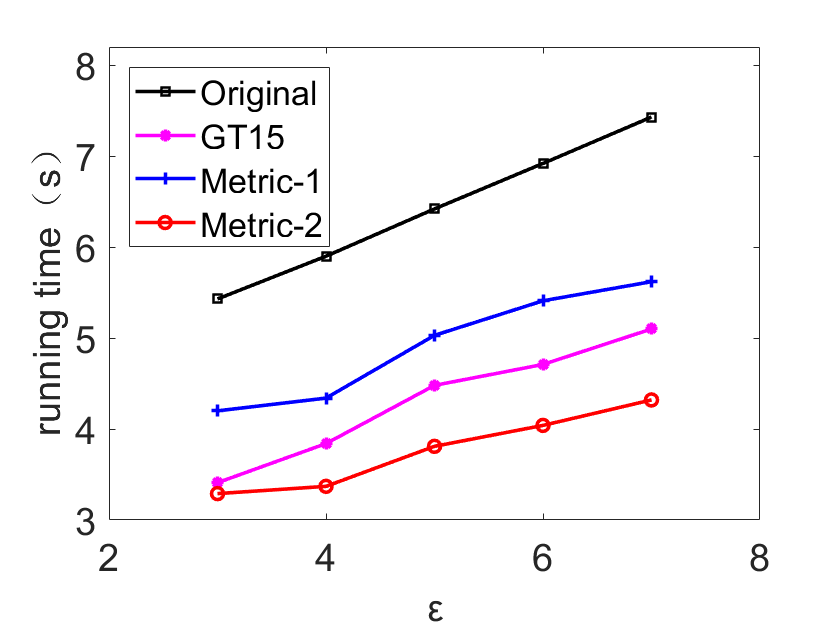}\label{fig-44}}
	\subfloat[\textsc{NeuroIPS} ($r=10$)]{\includegraphics[height=1.2in]{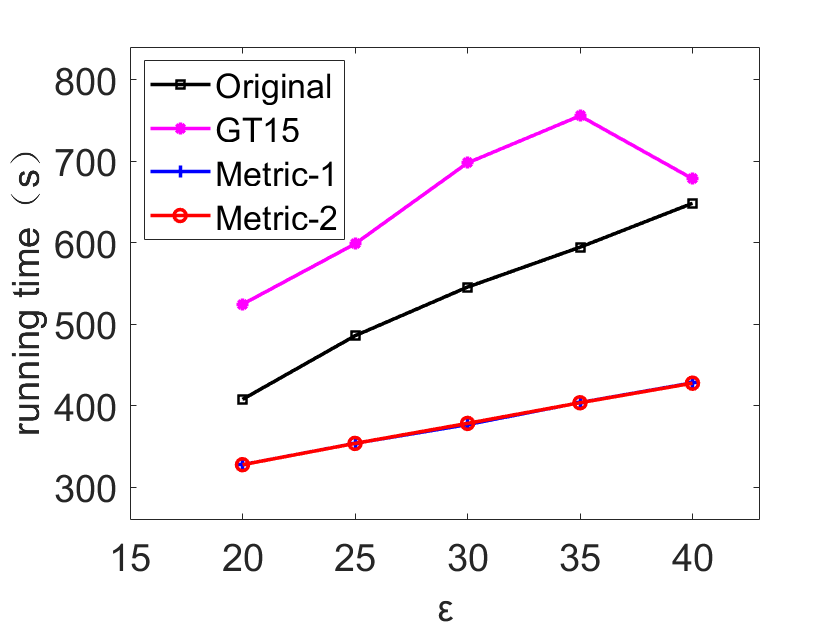}\label{fig-43}}
	\subfloat[\textsc{USPSHW} ($r=10$)]{\includegraphics[height=1.2in]{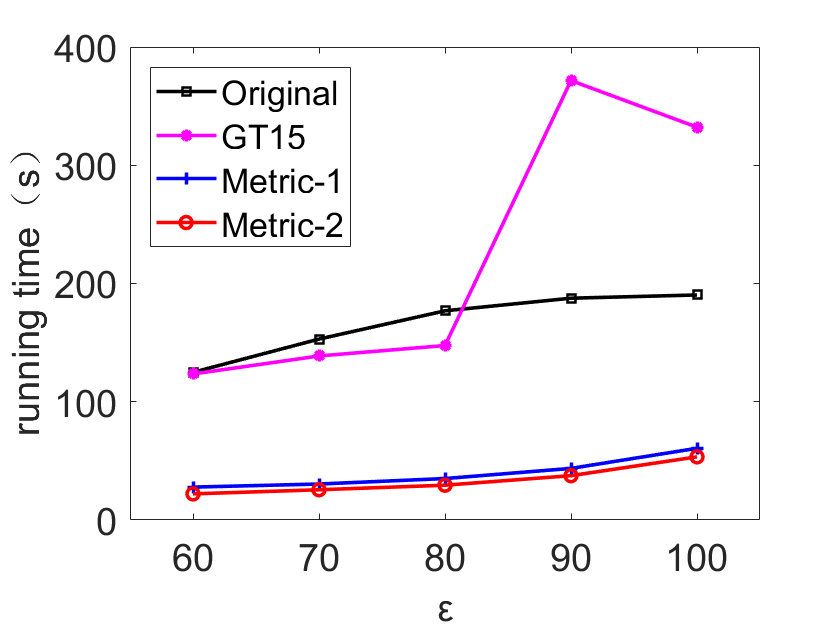} \label{fig-42}}
	\subfloat[\textsc{MINIST} ($r=15$)]{\includegraphics[height=1.2in]{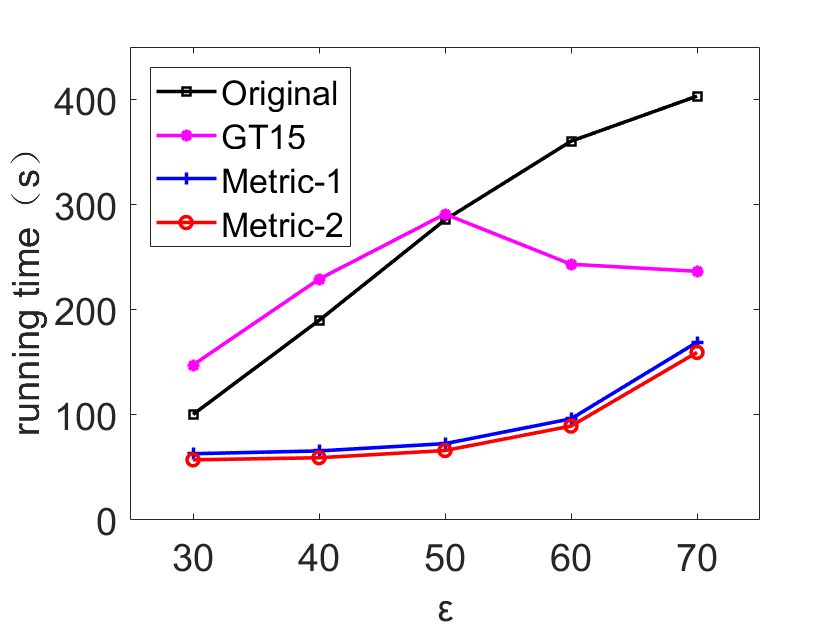}\label{fig-41}}\\
\vspace{-0.1in}
	\subfloat[\textsc{Synthetic} ($r=100$)]{\includegraphics[height=1.2in]{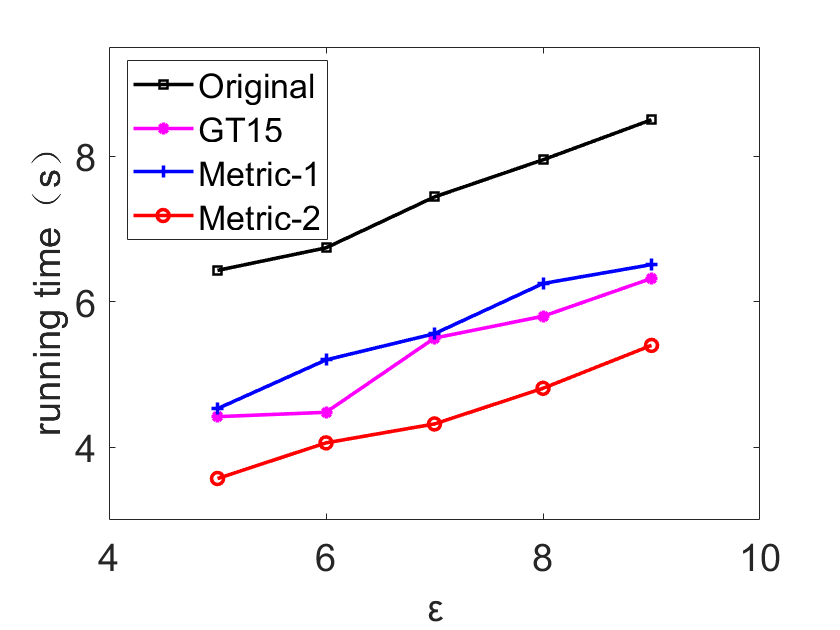}\label{fig-54}}
	\subfloat[\textsc{NeuroIPS} ($r=10$)]{\includegraphics[height=1.2in]{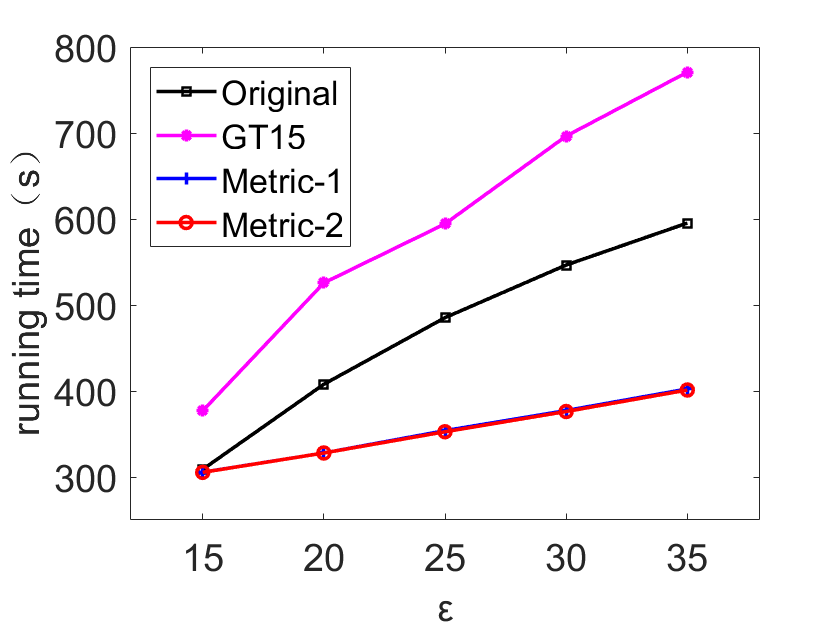}\label{fig-53}}
	\subfloat[\textsc{USPSHW} ($r=10$)]{\includegraphics[height=1.2in]{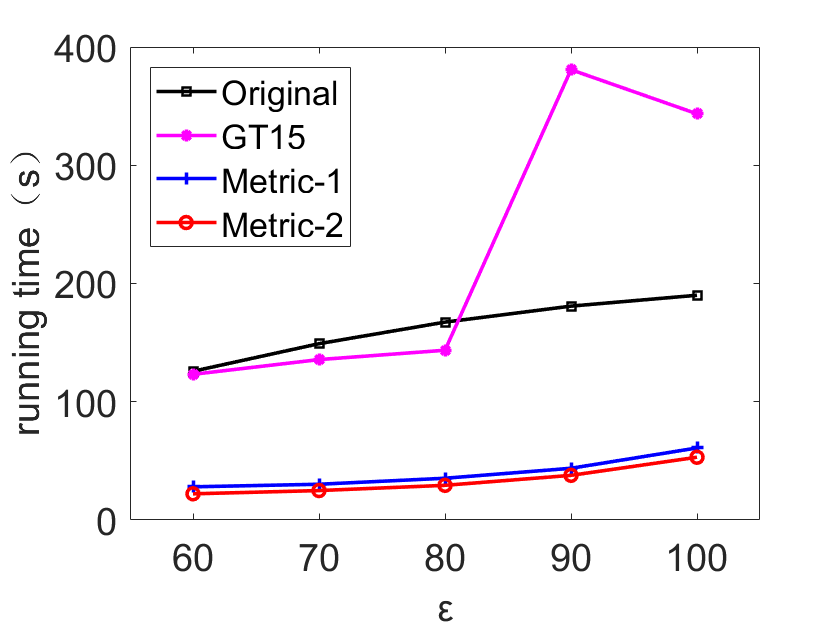} \label{fig-52}}
	\subfloat[\textsc{MINIST} ($r=15$)]{\includegraphics[height=1.2in]{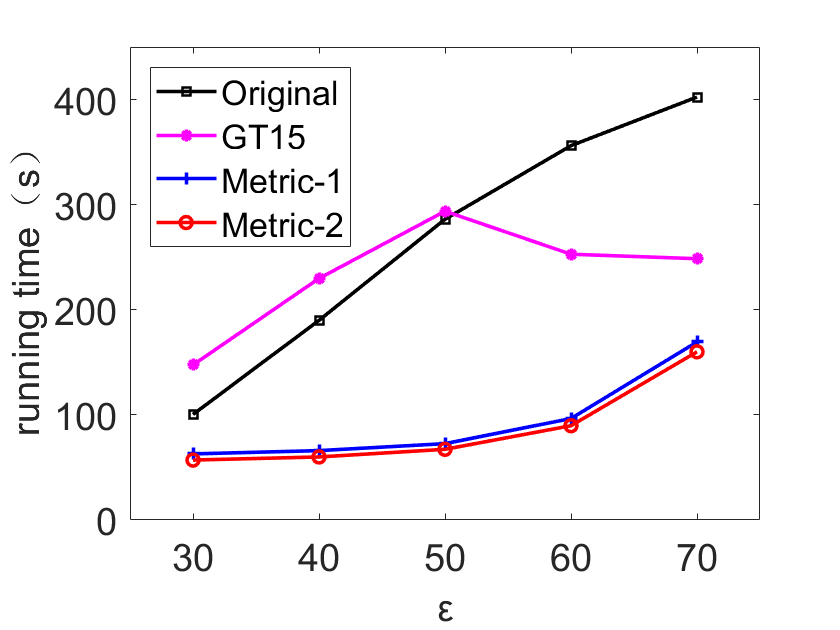}\label{fig-51}}\\
	\caption{Running times with  $MinPts=\frac{1}{1000}\cdot n$ (the first row) and $MinPts=\frac{2}{1000}\cdot n$ (the second row).}
	\label{fig-exp5}
\end{figure*}

\section{Experiments}
\label{sec-exp}
All the experimental results were obtained on a Windows $10$ workstation equipped with an Intel core $i5$-$8400$ processor and $8$GB RAM. 
We  compare the performances of the following four  DBSCAN algorithms in terms of running time: 

\begin{itemize}
	\item \textbf{\textsc{Original}}: the original DBSCAN \cite{ester1996density} that uses $R^*$-tree as the index structure. 
	\item \textbf{\textsc{GT15}}: the grid-based exact DBSCAN algorithm proposed in \cite{gan2015dbscan}. 
	\item \textbf{\textsc{Metric-1}}: our first DBSCAN algorithm proposed in Section~\ref{sec-dbscan1}.
	\item \textbf{\textsc{Metric-2}}: the alternative DBSCAN algorithm proposed in Section~\ref{sec-dbscan2}.
\end{itemize}
For the first two algorithms, we use the implementations in C++ from \cite{gan2015dbscan}. 
Our algorithms \textbf{\textsc{Metric-1}} and \textbf{\textsc{Metric-2}} are also implemented in C++.  Note that all of these four algorithms return the exact DBSCAN solution; we do not consider the approximate DBSCAN algorithms that are out of the scope of this paper. 


\begin{table}[]	
	\centering
	\caption{Datasets}	
	\label{tab:1}       
	\begin{tabular}{llll}		
		\hline\noalign{\smallskip}		
		Dataset & \#Instances & \#Attributes & Type  \\		
		\noalign{\smallskip}\hline\noalign{\smallskip}		
		\textsc{Synthetic} & $20000$ & $500$-$3000$ & Synthetic \\		
		\textsc{NeuroIPS} & $11463$ & $5811$ & Text \\	
		\textsc{USPSHW} & $7291$ & $256$ & Image	\\
		\textsc{MINIST} & $10000$ & $784$ & Image \\
		\noalign{\smallskip}\hline			
	\end{tabular}
	\label{tab-1}
\end{table}

\textbf{The datasets.} We evaluated our methods on both synthetic and real datasets where the details are shown in Table~\ref{tab-1}. We generated $6$ synthetic datasets. For each synthetic dataset, we randomly generate $n=20000$ points in $\mathbb{R}^2$, and then locate them to a higher dimensional space $\mathbb{R}^D$ through random affine transformations; the dimension $D$ ranges from $500$ to $3000$. 
\textbf{\textsc{NeuroIPS}}~\cite{perrone2017poisson} contains $n=11463$ word vectors of  the full texts of the NeuroIPS conference papers published in $1987$-$2015$. \textbf{\textsc{USPSHW}}~\cite{DBLP:journals/pami/Hull94} contains $n=7291$ $16\times16$ pixel handwritten letter images.  \textbf{\textsc{MNIST}}~\cite{lecun1998gradient} contains $n=10000$ handwritten digit images from $0$ to $9$, where each image is represented by a $784$-dimensional vector.

%
\begin{figure}[]
	\vspace{-0.1in}
	\centering
	\includegraphics[height=1.1in]{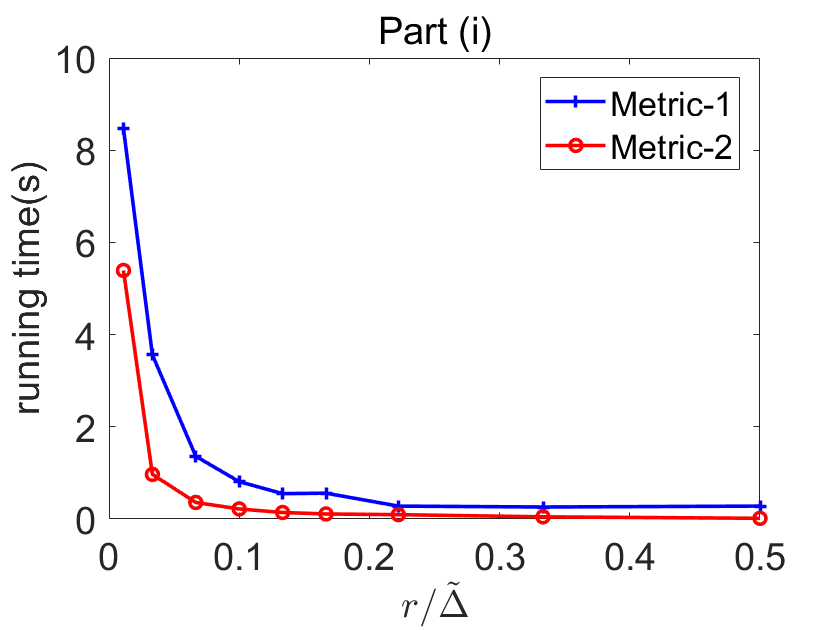}
		\includegraphics[height=1.1in]{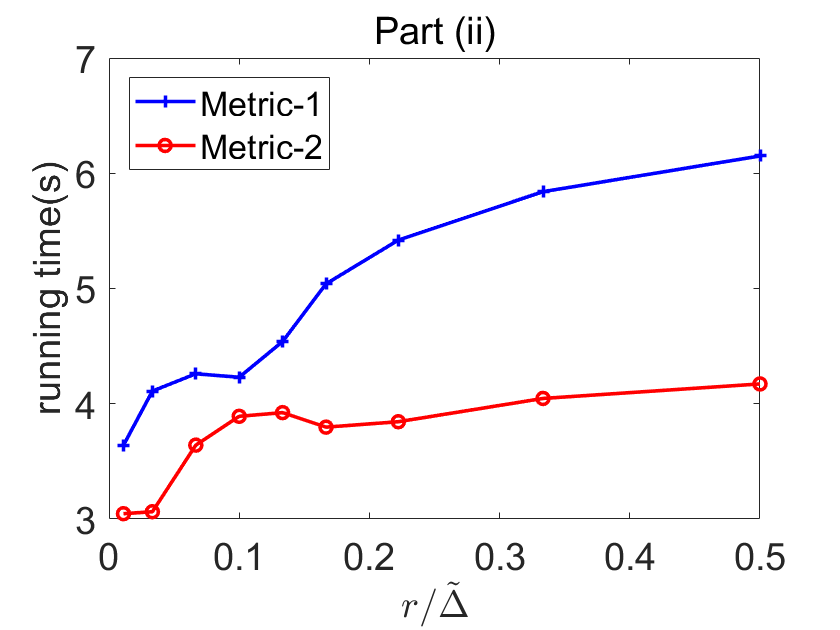}
	  \vspace{-0.1in}
	\caption{The running time in Part (\rmnum{1}) and Part (\rmnum{2}).} 
	\label{fig-exp1}
	  \vspace{-0.1in}
\end{figure}
%
%
%
%
%


%
%

\textbf{The results.} We validate the influence of the value of $r$ to the running times of \textsc{Metric-1} and \textsc{Metric-2}. We focus on the \textsc{Synthetic} datasets. To determine the value of $r$, we first estimate the diameter $\Delta$, the largest pairwise distance, of the dataset. Obviously, it takes at least quadratic time to achieve the exact value of $\Delta$; instead, we just arbitrarily select one point and pick its farthest point from the dataset, where the obtained value $\tilde{\Delta}$ is between $\Delta/2$ and $\Delta$. We set $\tilde{z}=200$ ({\em i.e.,} $1\% n$) and vary the ratio $r/\tilde{\Delta}$ in $0$-$0.5$. The running times with respect to  Part (\rmnum{1}) and Part (\rmnum{2}) (described in Section~\ref{sec-dbscan1}) are shown in Figure~\ref{fig-exp1} separately. As $r/\tilde{\Delta}$ increases, the running time of Part (\rmnum{1}) ({\em resp.,} Part~(\rmnum{2})) decreases ({\em resp.,} increases). The overall running time (of the two parts) reaches the lowest value when $r/\tilde{\Delta}$ is around $0.1$.

Further, we set the value $MinPts=\frac{1}{1000} n$ and $\frac{2}{1000}n$ for each dataset and show the running times in Figure~\ref{fig-exp5}. We can see that our \textsc{Metric-2} achieves the lowest running times on \textsc{Synthetic}; the running times of \textsc{Metric-1} and \textsc{Metric-2} are very close on the three real datasets; our both algorithms significantly outperform the two baseline algorithms in terms of running time.

\section{Future Work}
In this paper, we consider the problem of DBSCAN with low doubling dimension, and develop the $k$-center clustering based algorithms to reduce the complexity of range query. A few directions deserve to be studied in future work, such as
  other density based clustering and outlier recognition problems under the assumption of Definition~\ref{def-assumption}.

\newpage

\bibliographystyle{named}
\bibliography{dbscan}

\end{document}